\documentclass[runningheads]{llncs}

\usepackage[utf8]{inputenc}
\usepackage{amsmath,amssymb}
\usepackage{graphicx}
\usepackage{booktabs}
\usepackage{hyperref}
\usepackage{enumitem}
\usepackage{geometry}
\usepackage{subcaption}
\usepackage{adjustbox}

\begin{document}

\title{ Dual-Channel Closed Loop Supply Chain Competition: A Stackelberg--Nash Approach}
%
%
\author{ Gurkirat Wadhwa }
\authorrunning{Gurkirat et al.}
%
\institute{IEOR, IIT Bombay\\ 
\email{gurkirat@iitb.ac.in}}

\maketitle              

\title{\bfseries Dual-Channel Closed Loop Supply Chain Competition: A Stackelberg--Nash Approach}
\author{Gurkirat Wadhwa\\
Industrial Engineering and Operations Research (IEOR)\\
Indian Institute of Technology Bombay\\
\texttt{gurkirat@iitb.ac.in}}
\date{}

\begin{abstract}
In many consumer electronics and appliance markets, manufacturers sell products through
competing retailers while simultaneously relying on take-back programs to recover used items for re
manufacturing. Designing such programs is challenging when firms compete on prices and consumers
differ in their willingness to return products. Motivated by these settings, this paper develops a game
theoretic framework to analyze pricing and take-back decisions in a dual-channel closed-loop supply
chain (CLSC) with two competing manufacturers and two competing retailers. Manufacturers act as
Stackelberg leaders, simultaneously determining wholesale prices and consumer take-back bonuses,
while retailers engage in Nash competition over retail prices. The model integrates three key elements:
(i) segmented linear demand with cross-price effects, (ii) deterministic product returns, and (iii) an
inertia–responsiveness allocation mechanism governing the distribution of returned products between
manufacturers. Closed-form Nash equilibria are derived for the retailer subgame, along with symmetric Stackelberg equilibria for manufacturers. We derive a feasibility threshold for take-back incentives,
identifying conditions under which firms optimally offer positive bonuses to consumers. The results
further demonstrate that higher remanufacturing value or return rates lead the manufacturers to lower
wholesale prices in order to expand sales and capture additional return volumes, while high consumer
inertia weakens incentives for active collection. Numerical experiments illustrate and reinforce the analytical results, highlighting how consumer behavior, market structure and product substituitability
influence prices, bonuses, and return volumes. Overall, the study provides managerial insights for de
signing effective take-back programs and coordinating pricing decisions in competitive circular supply
chains.
\end{abstract}


\section{Introduction}
Growing environmental concerns and the rapid expansion of circular economy regulations are reshaping how firms design and manage supply chains (SCs). Beyond regulatory compliance, the collection and remanufacturing of used products has emerged as a strategic lever that can influence costs, prices, and competitive positioning. Prior research has shown that remanufacturing and product recovery can fundamentally alter pricing incentives and market outcomes in SCs \cite{Fleischmann2001,GuideVanWassenhove2009,Atasu2008}. However, when multiple firms simultaneously compete in both the forward product market and the reverse collection market, strategic tensions arise that complicate pricing, incentive provision, and profitability.

In a typical dual-channel closed-loop supply chain (CLSC), manufacturers sell new products to downstream retailers while also competing for the return of used products for remanufacturing or recycling. Forward-market interactions are often hierarchical, with manufacturers acting as leaders through wholesale pricing and retailers responding via retail prices, as extensively studied in the Stackelberg  literature \cite{JeulandShugan1983,Moorthy1988,CachonLariviere2005}. The reverse channel introduces an additional strategic layer: manufacturers must design consumer-facing take-back incentives to attract returns, thereby linking forward sales decisions to reverse collection outcomes. This coupling between forward and reverse markets implies that pricing and incentive decisions cannot be analyzed in isolation.

Despite substantial progress in CLSC modeling, most analytical frameworks adopt either centralized or cooperative decision-making assumptions, or focus on stochastic return processes that obscure the underlying strategic mechanisms \cite{Savaskan2003,Debo2005}. Game-theoretic studies of closed-loop systems have increasingly examined decentralized competition and environmental considerations \cite{Taleizadeh2020,Giri2021}, yet relatively few models jointly capture forward Stackelberg leadership and reverse-market competition with explicit consumer response to take-back incentives. As a result, important questions regarding incentive provision, equilibrium pricing, and collection efficiency under decentralized competition remain unanswered.

Another limitation of the existing analytical literature  is modeling of consumer return behavior. Empirical and behavioral studies suggest that consumer returns are heterogeneous: some consumers return products regardless of incentives due to habit, convenience, or institutional norms, while others are responsive to monetary bonuses \cite{Mukhopadhyay2018,Shulman2011}. Standard proportional allocation mechanisms commonly used in CLSC models don't capture this heterogeneity and frequently generate degenerate equilibria in which firms optimally set zero take-back incentives \cite{Debo2005}. Incorporating behavioral structure into return allocation is therefore essential for obtaining economically meaningful equilibrium predictions.

Motivated by these gaps, this paper develops a deterministic and analytically tractable Stackelberg--Nash model of a dual-channel CLSC with competing manufacturers and retailers. Manufacturers act as Stackelberg leaders, simultaneously setting wholesale prices and consumer take-back bonuses, while retailers engage in Nash competition over retail prices. Product returns are modeled deterministically and allocated across manufacturers using an inertia--responsiveness rule that captures both baseline returns and incentive-sensitive behavior. This structure helps us to derive closed-form subgame-perfect equilibria, an explicit feasibility threshold for positive take-back incentives, and comparative statics linking remanufacturing value and consumer responsiveness to equilibrium outcomes.

By combining pricing, decentralized competition, and behavioral return modeling within a unified analytical framework, the proposed model allows us to understand realistic market behaviour . The analytical results offer insights into when market-based take-back programs are self-sustaining and how profitability in the reverse channel feeds back into pricing decisions in the forward market. We also present a numerical study that illustrates and validates the closed-form equilibria. Using a symmetric baseline parameterization, the numerical analysis confirms that take-back incentives collapse to zero under proportional allocation, while a strictly positive bonus emerges under the proposed inertia--responsiveness rule. Additional numerical experiments highlight the threshold role of consumer responsiveness and inertia in incentive provision and demonstrate the monotonic impact of remanufacturing value on wholesale prices.

The remainder of the paper is organized as follows. Section~\ref{sec:model} describes the model and underlying assumptions. Section~\ref{sec:equil} derives the equilibrium outcomes and analytical results. Section~\ref{sec:numerical} presents the numerical illustrations and managerial insights, followed by concluding remarks and technical proofs in the appendices.

\subsection*{Literature Survey}
\label{sec:lit}
Research on vertical relations in supply chains (SCs) has long examined manufacturer–retailer interactions through Stackelberg leadership frameworks, with foundational studies analyzing wholesale and retail pricing, channel power, and coordination under linear demand systems \cite{JeulandShugan1983,Moorthy1988,CachonLariviere2005,Wadhwa2024ICORES}. These models are widely adopted due to their analytical tractability and well-behaved equilibria. The current work follows this tradition by adopting a manufacturer-led Stackelberg structure, while extending it to a circular supply chain setting in which upstream firms also compete strategically in the reverse market.

Closed-loop supply chain (CLSC) research originated in quantitative analyses of reverse logistics and product recovery. Early surveys and classification frameworks provided the foundations for analytical modeling of remanufacturing and returns management \cite{Fleischmann2001,GuideVanWassenhove2009}. Subsequent studies incorporated pricing and channel-design considerations, demonstrating how remanufacturing affects demand, profitability, and strategic interactions \cite{Savaskan2003,Atasu2008}. These models, however, often assume either a single manufacturer or centralized decision-making, thereby abstracting away from competition in the reverse channel. Authors in  \cite{LiDebo2009} analyze a CLSC with competing retailers, highlighting how downstream competition affects reverse-channel incentives. Building on this line of work, the present study explicitly models manufacturer-level competition in both forward and reverse markets, generating endogenous coupling between wholesale pricing and take-back decisions.

More recent contributions apply noncooperative game theory to CLSCs, employing Stackelberg and Nash equilibria to study remanufacturing competition, environmental investment, and regulatory intervention. Comprehensive reviews document the growing importance of game-theoretic approaches in this domain \cite{Taleizadeh2020,Giri2021}. Extensions examine sustainability investments \cite{Tang2016}, dual-channel pricing under emissions considerations \cite{Wang2018}, and cooperative versus competitive recycling under government policies \cite{Xie2022}. Despite their richness, many of these studies rely on stochastic demand, numerical equilibrium computation, or simulation-based methods. In contrast, the present study adopts a deterministic framework with linear demand with operating and non-operating decisions \cite{Wadhwa2024ICORES} and derives closed-form pricing equilibria in both forward and reverse markets, capturing strategic interactions realistically.

Another limitation of much of the analytical CLSC literature concerns the modeling of consumer return behavior. Empirical and behavioral studies show that return decisions are heterogeneous: some consumers return products regardless of incentives due to habit or institutional norms, while others are responsive to monetary take-back bonuses \cite{Mukhopadhyay2018,Shulman2011}. Standard proportional allocation rules commonly used in analytical models \cite{Debo2005} ignore this heterogeneity and frequently generate degenerate equilibria in which optimal bonuses collapse to zero. To address this issue, the present study incorporates an inertia–responsiveness mechanism that explicitly captures both baseline returns and incentive-sensitive behavior, ensuring interior equilibria and economically consistent comparative statics.

The present study builds on three interrelated streams of literature: vertical channel competition under hierarchical decision structures, analytical models of CLSC, and behavioral approaches to capture consumer product returns. We now proceed with the detailed problem description. 
\section{Model Description}
\label{sec:model}

We consider a deterministic dual-channel closed-loop supply chain (CLSC) consisting of two manufacturers and two retailers. 
The overall structure of the system, including the direction of material flows, information flows, and the timing of strategic decisions, is illustrated in Figure~\ref{fig:model}.

\begin{figure}[h]
\centering
\includegraphics[width=0.9\linewidth]{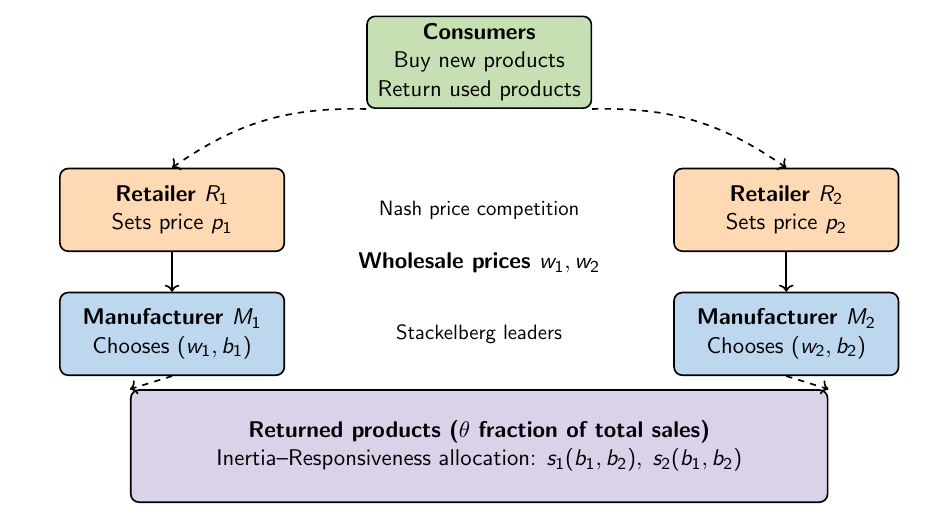}
\caption{Dual-channel CLSC Model}
\label{fig:model}
\end{figure}

As shown in Figure~\ref{fig:model}, manufacturers \(M_1\) and \(M_2\) operate upstream and supply new products to retailers \(R_1\) and \(R_2\), respectively. 
Retailers compete downstream by setting retail prices and selling the products to a common consumer market. 
Consumers purchase new products from retailers and subsequently generate used products that may be returned for recovery, giving rise to a reverse channel alongside the forward sales channel.

The SC agents' interaction is modeled as a three-stage Stackelberg--Nash game. 
In the first stage, manufacturers act as Stackelberg leaders and simultaneously choose wholesale prices and consumer take-back bonuses. 
In the second stage, retailers observe these decisions and engage in Nash competition by setting retail prices. 
In the third stage, a deterministic fraction of sold products is returned by consumers and allocated between manufacturers based on the take-back incentives offered, after which payoffs are realized.

Figure~\ref{fig:model} illustrates the coupling between the forward and reverse channels. 
Wholesale and retail pricing decisions determine sales volumes, which in turn govern the total quantity of returned products. 
At the same time, take-back incentives influence how these returns are distributed between manufacturers, directly affecting their remanufacturing profits. 
This feedback structure creates strategic interdependence among pricing and return decisions, which is central to the analysis. We now explain each component of the model one after the another.

\subsection{Players, timing, and decisions}

The interaction is modeled as a three-stage Stackelberg--Nash game, as illustrated in Figure~\ref{fig:model}. 
Manufacturers $M_1$ and $M_2$ act as Stackelberg leaders supplying products to retailers $R_1$ and $R_2$ respectively who act as followers. 
At each level of the hierarchy, firms interact non-cooperatively in a Nash framework.

\begin{itemize}
    \item[i)]\textit{Stage 1 (Manufacturers -- leaders):}
        Each manufacturer $M_i$, $i\in \{1,2\}$  simultaneously chooses wholesale prices $w_i$ and consumer take-back bonuses $b_i$. 
    At this stage, manufacturers anticipate retailers’ pricing responses and the resulting return allocations.
    The wholesale price $w_i$ is charged to retailer $R_i$ for each unit sold, while the bonus $b_i$ is paid directly to consumers for each returned unit collected by  $M_i$.

    \item[ii)] \textit{Stage 2 (Retailers -- followers):}
    After observing $(w_1,w_2,b_1,b_2)$, retailers $R_1$ and $R_2$ simultaneously choose retail prices $p_1$ and $p_2$.
    Retailers compete in prices in a Nash equilibrium, taking manufacturers’ decisions as given.

    \item[iii)] \textit{Stage 3 (Returns and payoffs):}
    Consumer demand is realized, a deterministic fraction of sold products is returned and allocated between manufacturers according to the take-back incentives, and payoffs for all firms are realized.
\end{itemize}

This timing reflects the hierarchical nature of the SC; upstream manufacturers commit to wholesale prices and take-back programs first, while downstream retailers subsequently adjust prices competitively in response.

\subsection{Demand structure}

Consumers perceive the two retailers as offering differentiated but substitutable products. 
Following standard practice in competitive pricing models, the demand faced by retailer $i$ is given by (see \cite{Wadhwa2024ICORES})
\begin{equation}
D_i(p_i,p_j)
=
\bigl(\bar d_i - \alpha_i p_i + \varepsilon \alpha_j p_j\bigr)^+,
\qquad i \neq j,
\label{eq:demand}
\end{equation}
where $\bar d_i > 0$ denotes the potential market size of retailer $i$, $\alpha_i > 0 $ captures consumers’ sensitivity to retailer $i$’s quoted price, and $\varepsilon \in [0,1)$ measures the degree of substituitability between the two retailers’ products.
The operator $(\cdot)^+$ ensures non-negativity of demand. 
We focus on interior equilibria, so this constraint is inactive in the subsequent analysis.

\textit{Economic interpretation:}
The term $(\bar d_i - \alpha_i p_i)$ represents the standard downward-sloping demand relationship, where demand decreases with the retailer’s quoted price and is sensitized by $\alpha_i$. A larger value of $\alpha_i$ corresponds to more price-sensitive consumers and more aggressive price competition, whereas a smaller value implies relatively inelastic demand. 
The cross-price term $\varepsilon \alpha_j p_j$ captures demand diversion arising from competitive interactions; an increase in the rival retailer’s price shifts a portion of consumers toward retailer $i$.
The parameter $\varepsilon$ can be interpreted as a measure of consumer switching propensity. 
When $\varepsilon$ is close to \textit{zero}, products are weak substitutes and consumers exhibit strong loyalty to a given retailer; demand is therefore primarily driven by each retailer’s own pricing decision (see \cite{Wadhwa2024ICORES}) . 
As $\varepsilon$ increases toward \textit{one}, products become increasingly substitutable and interchangeable, and consumers are more willing to switch between retailers in response to relative price differences. 
This captures markets in which consumers primarily seek product availability rather than retailer-specific attributes.
This linear demand model with cross-price effects is widely adopted in differentiated duopoly and SC competition models because it offers analytical tractability while preserving economically meaningful competitive interactions (e.g., \cite{SinghVives1984,Wadhwa2024ICORES}).

\subsection{Return generation and allocation mechanism}

A key feature of the model is the treatment of product returns in the reverse channel. 
Consistent with empirical and modeling evidence in CLSCs (e.g., \cite{Fleischmann2001,Atasu2013,GuideVanWassenhove2009}), we assume that a deterministic fraction $\theta \in (0,1)$ of total sales is returned by consumers after use. 
Accordingly, the total quantity of returned products is given by
\begin{equation}
Q_{\text{tot}} = \theta (D_1 + D_2).
\label{eq:totalreturns}
\end{equation}

Returned products are then allocated between manufacturers based on the take-back incentives they offer. 
Following the behavioral literature on consumer return decisions, we adopt an inertia--responsiveness allocation rule:
\begin{equation}
s_i(b_1,b_2)
=
\frac{\beta_i b_i + \gamma_r}
{\beta_1 b_1 + \beta_2 b_2 + 2\gamma_r},
\qquad i=\{1,2\},
\label{eq:allocation}
\end{equation}
where $\beta_i \ge 0$ measures the responsiveness of consumers to manufacturer $i$’s bonus, and $\gamma_r > 0$ captures baseline inertial returns that are independent of monetary incentives. The inertia–responsiveness allocation rule in \eqref{eq:allocation} abstracts from channel-specific frictions such as the retailer from which the product was purchased, differences in return convenience across channels, or retailer-managed collection programs. Instead, returns are modeled as being directed solely by manufacturer-level take-back incentives and baseline consumer inertia. This reduced-form approach allows us to isolate strategic competition among manufacturers in the reverse channel and is analytically tractable. Incorporating retailer-dependent return frictions or channel-specific collection costs would be an interesting future direction.

Under this rule, manufacturer $i$ collects
\begin{equation}
Q_{r,i} = \theta (D_1 + D_2) s_i(b_1,b_2)
\end{equation}
units for remanufacturing.
The allocation mechanism reflects heterogeneity in consumer return behavior. 
A segment of consumers returns used products regardless of incentives due to habit, convenience, or institutional arrangements such as default collection channels; this behavior is captured by the inertia parameter $\gamma_r$. 
Another segment responds to monetary take-back incentives, with sensitivity measured by $\beta_i$.

When $\gamma_r = 0$, the allocation rule reduces to the proportional form
$s_i = b_i / (b_1 + b_2)$, which has been widely adopted in competitive take-back models (e.g., \cite{Savaskan2003}). 
However, in noncooperative settings this proportional formulation often leads to corner solutions in which all firms optimally choose zero bonuses, resulting in degenerate equilibria.

Introducing baseline inertia ensures that returns occur even when incentives are weak, guarantees interior allocation shares, and preserves analytical tractability and allows us to derive economically meaningful equilibria while remaining consistent with observed consumer return behavior documented in empirical studies (e.g., \cite{Mukhopadhyay2018,Shulman2011}).

\subsection{Costs, profits, and operating decisions}

Each retailer $R_i$ and each manufacturer $M_i$ incurs a fixed operating cost, denoted by $O_{R_i}$ and $O_{M_i}$, respectively. 
 Firms participate in the market only if their equilibrium profit is nonnegative (see \cite{Wadhwa2024ICORES}) ; otherwise, they optimally choose not to operate and earn zero payoff.

If retailer $R_i$ operates, its profit is given by
\begin{equation}
\pi_{R_i}(p_i,p_j;w_i)
=
(p_i - w_i) D_i(p_i,p_j) - O_{R_i},
\label{eq:retprofit}
\end{equation}
where the first term represents the retail margin earned on each unit sold and the second term captures the fixed cost of operation.

If manufacturer $M_i$ operates, its profit is given by
\begin{equation}
\pi_{M_i}
=
(w_i - c_i) D_i^*
+
(v_i - b_i - k_i) Q_{r,i}
-
O_{M_i},
\label{eq:manprofit}
\end{equation}
where $c_i$ denotes the unit production cost of new products, $k_i$ is the per-unit processing cost associated with returned products (excluding consumer bonuses), $v_i$ represents the remanufacturing or salvage value obtained from each returned unit, and $D_i^*$ denotes equilibrium demand after retailers optimally set prices.

Observe that the manufacturer’s profit function \eqref{eq:manprofit} consists of two distinct components. 
The first term captures the forward-channel margin generated from wholesale sales to retailers. 
The second term captures the reverse-channel margin associated with remanufacturing returned products. 
The take-back bonus $b_i$ directly reduces the net value extracted from each returned unit, while indirectly affecting manufacturer profit by influencing the allocation of returns across competing manufacturers.  We immediately deduce the following result.

\begin{lemma} [Reverse-channel viability]
\label{lem:reverse_viability}
If the net remanufacturing value satisfies $v_i-k_i\le 0$, then in any equilibrium of the game manufacturers optimally set zero take-back bonuses, i.e.,
\[
b_i^* = 0 \quad \text{for all } i \in \{1,2\},
\]
and do not actively collect returned products. Consequently, the reverse channel is inactive and the model reduces to a forward-only Stackelberg pricing game.
\end{lemma}

\noindent \textit{Proof:}
From the manufacturer profit function in \eqref{eq:manprofit}, the reverse-channel payoff of manufacturer $i$ is
\[
\pi_{M_i}^{\emph{rev}}(b_i,b_j)
= (v_i - b_i - k_i)\, Q_{r,i}(b_i,b_j),
\]
where $Q_{r,i}(b_i,b_j)\ge 0$ denotes the quantity of returned products collected by manufacturer $i$ and the take-back bonus satisfies the feasibility constraint $b_i\ge 0$.

Fix the rival’s decision $b_j\ge 0$. Suppose that $v_i-k_i\le 0$. Then for any feasible $b_i\ge 0$,
\[
v_i - b_i - k_i \le v_i - k_i \le 0.
\]
Since $Q_{r,i}(b_i,b_j)\ge 0$ for all $(b_i,b_j)$, it follows that
\[
\pi_{M_i}^{\emph{rev}}(b_i,b_j)\le 0
\quad \text{for all } b_i\ge 0.
\]

To characterize the optimal bonus choice, differentiate $\pi_{M_i}^{\emph{rev}}$ with respect to $b_i$:
\[
\frac{\partial \pi_{M_i}^{\emph{rev}}}{\partial b_i}
= -Q_{r,i}(b_i,b_j)
+ (v_i-b_i-k_i)\frac{\partial Q_{r,i}}{\partial b_i}.
\]
Under the maintained model assumptions, $\frac{\partial Q_{r,i}}{\partial b_i}\ge 0$, and when $v_i-k_i\le 0$ we have $v_i-b_i-k_i\le 0$ for all $b_i\ge 0$. Hence,
\[
\frac{\partial \pi_{M_i}^{\emph{rev}}}{\partial b_i}\le 0
\quad \text{for all } b_i\ge 0,
\]
with strict inequality whenever $Q_{r,i}>0$ or $\frac{\partial Q_{r,i}}{\partial b_i}>0$.

Therefore, $\pi_{M_i}^{\emph{rev}}(b_i,b_j)$ is weakly decreasing in $b_i$ over the feasible set $b_i\ge 0$. By the Kuhn–Tucker optimality conditions for the constrained maximization problem, the optimal solution is attained at the boundary $b_i=0$.
Setting $b_i=0$ eliminates any positive marginal incentive to attract returned products, and any equilibrium allocation of returns yields zero or negative reverse-channel profit. Hence, in equilibrium, manufacturer $i$ optimally chooses
\[
b_i^* = 0.
\]
Since this argument is applicable to both manufacturers, the reverse channel is inactive in equilibrium.
\qed

\textit{Remark:}
Lemma~\ref{lem:reverse_viability} establishes a necessary viability condition for the CLSC. When the net remanufacturing margin for manufacturer $M_i$, $v_i-k_i$ is non-positive, returned products destroy value before incentive payments, eliminating manufacturers' incentives to collect them. Active take-back programs and forward--reverse market coupling arise only when $v_i-k_i>0$ for $i= \{1,2\}$, which is therefore assumed throughout the remainder of the analysis.

\subsection{Assumptions and market structure}

To ensure analytical tractability and consistency with the proposed market structure, we impose the following assumptions throughout the paper.

\noindent\textbf{(A.1)} Manufacturers are responsible for designing consumer take-back programs and processing returned products. Retailers do not participate in reverse logistics and have no decision variables related to product returns.\\
\noindent\textbf{(A.2)} The market potential $\bar d_i$ for each retailer $R_i$ is sufficiently large to ensure strictly positive equilibrium demand under the interior solutions considered in the analysis.\\
\noindent\textbf{(A.3)} Each firm has the option to operate or not-operate. If a firm’s equilibrium profit is negative, it optimally chooses not to operate and earns zero payoff.

Throughout the analysis, we focus on interior wholesale pricing equilibria in which manufacturers earn nonnegative margins from forward sales. Although it is possible for reverse-channel profits to  aggressive wholesale pricing, \textbf{(A.3)} rules out equilibria in which firms operate at a loss. This restriction ensures economically meaningful equilibria and allows us to characterize closed-form pricing solutions. Examining loss-leading strategies supported by reverse-channel profits would require relaxing the operating decision assumption and is left for future research.

In all, the developed model formulation defines a deterministic and analytically tractable framework that integrates forward pricing competition with reverse-channel incentive provision under decentralized decision-making. This allows us to characterize equilibrium pricing, take-back incentives, and their interaction in competitive circular SCs. We now proceed with the equilibrium analysis of the model.

\section{Equilibrium Analysis}\label{sec:equil}
This section characterizes the subgame-perfect equilibrium (SPE) of the three-stage Stackelberg--Nash game using backward induction. 
We begin by solving the retailers’ pricing subgame for given manufacturer decisions. 
We then derive manufacturers’ optimal choices anticipating the retailers’ equilibrium responses, and finally examine the symmetric equilibrium and its implications.

\subsection{Retailer Nash Equilibrium}

Given wholesale prices $\mathbf{w}=(w_1,w_2)$ and take-back bonuses $\mathbf{b}=(b_1,b_2)$ determined by manufacturers in Stage~1, each retailer $R_i$ chooses its retail price $p_i$ to maximize profit:
\begin{equation}
\max_{p_i \ge 0} \; \pi_{R_i}(p_i,p_j;w_i)
= (p_i - w_i)\bigl(\bar d_i - \alpha_i p_i + \varepsilon \alpha_j p_j\bigr) - O_{R_i}.
\end{equation}
We focus on interior solutions. The first-order condition for retailer $i$ is therefore given by
\begin{equation}
\bar d_i + \alpha_i w_i + \varepsilon \alpha_j p_j - 2 \alpha_i p_i = 0.
\tag{R-FOC}
\end{equation}

Solving the system of first-order conditions yields the retailers’ Nash equilibrium prices.

\begin{theorem}[Retailer-stage uniqueness]\label{thm:retailer}
For $\alpha_i>0$ and $\varepsilon \in [0,1)$, the retailer-stage pricing game admits a unique Nash equilibrium. 
The equilibrium prices are given by
\begin{align}
p_1^*(w_1,w_2) &= 
\frac{2(\bar d_1 + \alpha_1 w_1) + \varepsilon(\bar d_2 + \alpha_2 w_2)}
{\alpha_1(4-\varepsilon^2)}, \label{eq:p1star_main}\\[3pt]
p_2^*(w_1,w_2) &= 
\frac{\varepsilon(\bar d_1 + \alpha_1 w_1) + 2(\bar d_2 + \alpha_2 w_2)}
{\alpha_2(4-\varepsilon^2)}. \label{eq:p2star_main}
\end{align}
\end{theorem}

\noindent\textit{Proof:} See Appendix~\ref{app:retailer_proof}. \qed

\textit{Remark:}
Equations~\eqref{eq:p1star_main}--\eqref{eq:p2star_main} characterize retailers’ optimal pricing responses to manufacturers’ wholesale prices. 
Each retailer’s equilibrium price is increasing in its own manufacturer's wholesale price $w_i$, reflecting pass-through of upstream cost changes determined by the manufacturer. 
Moreover, a higher wholesale price charged to the rival retailer, $w_j$, increases retailer $i$’s equilibrium price through reduced competitive pressure in the downstream market, with the magnitude of this effect governed by the substituitability parameter $\varepsilon$. 
When $\varepsilon \approx 0$, retailers operate in effectively segmented markets and exhibit negligible strategic interaction; whereas when $\varepsilon \approx 1$, it represents stronger price interdependence and greater demand diversion across retailers.

\subsection{Manufacturers’ Leader-Stage Conditions}
To characterize manufacturers’ optimal decisions, we substitute the equilibrium retail prices $p_i^*(\mathbf{w})$ for each $i=\{1,2\}$ obtained from Theorem \ref{thm:retailer} into demand functions \eqref{eq:demand} yields linearized demands $D_i^*(w_1,w_2)$ ( see Appendix \ref{app:sub}). 
Each manufacturer $M_i$ then solves:
\[
\max_{w_i,b_i \ge 0} \; \pi_{M_i}(\mathbf{w},\mathbf{b})
= (w_i - c_i)D_i^*(w_1,w_2)
+ (v_i - b_i - k_i)\theta\!\!\sum_{k=1}^2 D_k^*(w_1,w_2) s_i(b_1,b_2)
- O_{M_i}.
\]
The interior first-order conditions are:
\begin{align}
0 &= D_i^* + (w_i - c_i)\frac{\partial D_i^*}{\partial w_i}
+ (v_i - b_i - k_i)\,\theta\sum_{k=1}^2 
\frac{\partial D_k^*}{\partial w_i}\, s_i(b_1,b_2), \label{eqn_1} \\[2pt]
0 &= -s_i(b_1,b_2) + (v_i - b_i - k_i)\frac{\partial s_i}{\partial b_i}. \label{eqn_2}
\end{align}

\textit{Remark:}
Equation \eqref{eqn_1} balances the marginal gain from higher wholesale price (first term) against the induced loss in retailer demand (second term) and the indirect effect on reverse-channel value through $\theta(v_i-b_i-k_i)$.
Equation \eqref{eqn_2} equates the marginal benefit of a higher take-back bonus (increasing returns) with its marginal cost. 
Together, these form the manufacturers’ best-response system.

\subsection{Symmetric Equilibrium}

To obtain analytical insights and understand the structural effects of competition and return incentives, we focus on a symmetric benchmark in which manufacturers and retailers are ex ante identical. 
Specifically, we assume
\[
\bar d_1=\bar d_2=\bar d,\quad 
\alpha_1=\alpha_2=\alpha,\quad 
c_1=c_2=c,\quad 
v_1=v_2=v,\quad 
k_1=k_2=k,\quad 
\beta_1=\beta_2=\beta,
\]
and $O_{M_i}=O_M$, $O_{R_i}=O_R$. 
We therefore seek a symmetric equilibrium in which manufacturers choose identical wholesale prices and take-back bonuses, i.e.,
\[
(w_1,w_2,b_1,b_2)=(w,w,b,b).
\]

Under symmetry, the retailers’ equilibrium prices from Theorem ~\ref{thm:retailer} and demands (obtained by substituting prices in \eqref{eq:demand}) simplify to
\begin{eqnarray}\label{eqn_sym_pr_dem}
 p = \frac{\bar d/\alpha + w}{2 - \varepsilon}, 
\qquad
D(w) = \frac{\bar d - \alpha(1-\varepsilon) w}{2 - \varepsilon}.   
\end{eqnarray}

\begin{theorem}[Symmetric equilibrium and bonus threshold]\label{thm:symmetric}
Under the symmetric setting described above, the equilibrium outcomes are characterized as follows.
\begin{enumerate}[label=(\alph*),noitemsep]
  \item \textbf{Proportional allocation.} 
  If returned products are allocated proportionally or if consumer bonuses do not affect return behavior, the equilibrium take-back bonus is $b^*=0$, and the corresponding wholesale price is
  \[
  w^* = \frac{\bar d}{2\alpha(1-\varepsilon)} + 
  \frac{c - \theta(v - k)}{2}.
  \]

  \item \textbf{Inertia--responsiveness allocation.} 
  Under the inertia--responsiveness allocation rule \eqref{eq:allocation}, the equilibrium take-back bonus is given by
  \[
  b^* = \frac{v - k}{3} - \frac{2\gamma_r}{3\beta},
  \]
  which is feasible if and only if $\beta(v-k) > 2\gamma_r$. 
  The corresponding equilibrium wholesale price is
  \[
  w^* = \frac{\bar d}{2\alpha(1-\varepsilon)} 
  + \frac{c - \theta(v - b^* - k)}{2}.
  \]
\end{enumerate}
\end{theorem}

\noindent\textit{Proof:} See Appendix~\ref{app:symmetric_proof}. \qed

\textit{Remark:}
Part~(a) establishes a benchmark result; when take-back incentives do not influence the allocation of returned products, manufacturers optimally set zero bonuses, reproducing the classical Stackelberg outcome observed in forward-only supply chains. 
Part~(b) shows how introducing the inertia--responsiveness mechanism fundamentally alters equilibrium behavior. 
Higher consumer responsiveness to bonuses (larger $\beta$) or lower baseline inertia (smaller $\gamma_r$) increases the equilibrium take-back bonus, as manufacturers compete to attract returned products.
The feasibility condition $\beta(v-k) > 2\gamma_r$ provides an explicit threshold for the emergence of active take-back incentives. 
When this condition is not satisfied, inertia dominates responsiveness and manufacturers optimally refrain from offering bonuses. 

It is important to observe that in the symmetric equilibrium, each manufacturer receives an equal share of returned products, 
$s_i= \frac{1}{2}$ for each $i$,(see Appendix \ref{app:symmetric_proof}) independent of the bonus level. Nevertheless, take-back incentives might remain strictly positive because bonuses affect the marginal gain from unilateral deviations. A marginal increase in $b_i$
 raises manufacturer 
$i’s$ return share relative to its rival , even though equilibrium shares coincide under symmetry. As a result, each manufacturer competes for returns at the margin, which sustains strictly positive equilibrium bonuses despite equal allocation shares in equilibrium.

\begin{theorem}[Forward--reverse coupling under proportional allocation]
\label{thm:cross_market_proportional}
Consider the symmetric equilibrium under proportional return allocation described in Theorem \ref{thm:symmetric}(a), and suppose that the net remanufacturing margin satisfies $v-k>0$. Holding all other parameters fixed, the equilibrium wholesale price $w^*$ satisfies
\[
\frac{\partial w^*}{\partial v} = -\frac{\theta}{2} < 0,
\qquad
\frac{\partial w^*}{\partial \theta} = -\frac{v-k}{2} < 0.
\]
Hence, an increase in either the remanufacturing value $v$ or the return fraction $\theta$ strictly reduces the equilibrium wholesale price.
\end{theorem}

\noindent\textit{Proof:}
See Appendix \ref{app:cross_market_proof}.
\qed

\textit{Remark:}
Theorem~\ref{thm:cross_market_proportional} formalizes a forward--reverse coupling effect in closed-loop supply chains. An increase in remanufacturing profitability, arising either from higher per-unit recovery value or a higher return rate, raises the marginal value of expanding total sales. Anticipating additional returns generated by greater sales volumes, manufacturers optimally reduce wholesale prices to stimulate downstream demand. This mechanism operates even when return allocation is proportional and take-back incentives are absent, highlighting that reverse-channel profitability can influence forward pricing decisions independently of active bonus competition.

\begin{corollary}[Forward--reverse coupling under inertia--responsiveness]
\label{cor:coupling_inertia}
Consider the symmetric equilibrium under the inertia--responsiveness allocation rule described in Theorem~\ref{thm:symmetric}(b), and assume the interior regime $\beta(v-k)>2\gamma_r$. Holding all other parameters fixed, the equilibrium wholesale price $w^*$ satisfies
\[
\frac{\partial w^*}{\partial v} = -\frac{\theta}{3} < 0,
\qquad
\frac{\partial w^*}{\partial \theta} = -\frac{v-b^*-k}{2} < 0.
\]
Hence, increases in either the remanufacturing value $v$ or the return fraction $\theta$ strictly reduce the equilibrium wholesale price.
\end{corollary}

\noindent\textit{Proof:}
See Appendix \ref{app_coupling}.
\qed

\textit{Remark:}
Corollary~\ref{cor:coupling_inertia}, together with Theorem~\ref{thm:cross_market_proportional}, shows that the forward--reverse coupling between wholesale pricing and remanufacturing profitability is robust to the choice of return-allocation mechanism. In both proportional and inertia--responsiveness settings, higher remanufacturing value or return rates induce manufacturers to lower wholesale prices in order to expand sales and capture additional returns. However, consumer inertia attenuates this response (as $|\frac{\partial w^{*}}{\partial v}| = \frac{\theta}{3} < \frac{\theta}{2} $ ); part of the incremental remanufacturing value is absorbed through higher take-back bonuses rather than being fully passed through into lower wholesale prices, resulting in a weaker price response than under proportional allocation.

\subsection{Welfare Comparison}

We conclude the analytical analysis by comparing the decentralized Stackelberg--Nash equilibrium with a socially optimal benchmark. 
The objective of this comparison is to determine whether decentralized pricing and take-back decisions internalize the full social value created by remanufacturing and product recovery.

Consider a benevolent social planner who chooses pricing and incentive variables to maximize total surplus. 
Total surplus is defined as the sum of consumer surplus, firms’ profits in both the forward and reverse channels, and the net remanufacturing value generated by returned products. 
Since take-back bonuses represent pure monetary transfers between consumers and manufacturers, they do not affect total surplus and therefore are omitted from the social planner’s decision variables. The planner internalizes the full remanufacturing value $\theta v$ generated by returns, but does not need to explicitly choose bonus levels.
In particular, the planner recognizes that each unit sold generates an expected social benefit of $\theta v$ through subsequent recovery.

In contrast, under decentralized decision-making, each manufacturer internalizes only the private net remanufacturing value associated with returned products, namely $\theta(v-b-k)$ per unit sold, since take-back bonuses and processing costs are privately borne. 
This difference implies that decentralized manufacturers do not fully internalize the positive externality created by forward sales through the reverse channel.

Let $D(w)$ denote equilibrium demand as a function of the wholesale price in \eqref{eqn_sym_pr_dem}).
Under decentralized competition and symmetric case, manufacturers choose prices and incentives to satisfy the private optimality condition (see \eqref{eq:manprofit} and $s_i = \frac{1}{2}$ from \eqref{eq:allocation} in symmetric case). 
\begin{eqnarray}\label{eqn_welfare_1}
    D(w)+\big[(w-c)+\theta(v-b-k)\big]D'(w)=0.
\end{eqnarray}
Under the social planner’s problem, the corresponding first-order condition is
\begin{eqnarray}\label{eqn_welfare_2}
  D(w)+\big[(w-c)+\theta v\big]D'(w)=0.  
\end{eqnarray}

\begin{theorem}[Welfare inefficiency]
\label{thm:welfare}
For any $\theta v>0$, the wholesale price  chosen under
the decentralized equilibrium $w^{*}_{D}$ is strictly higher than the wholesale price that
maximizes total surplus under the social planner $w^{*}_{SP}$. Consequently,
equilibrium market demand and total return volumes are strictly lower under the
decentralized equilibrium than under the socially optimal outcome.
\end{theorem}

\noindent\textit{Proof:} See Appendix~\ref{app:welfare_proof}. \qed

\textit{Remark:}
Theorem~\ref{thm:welfare} identifies a fundamental inefficiency arising from decentralized decision-making in CLSC. 
While remanufacturing generates additional social value through resource conservation and waste reduction, individual manufacturers only partially internalize this value when setting wholesale prices and take-back bonuses. 
As a result, equilibrium prices remain inefficiently high and incentives for product recovery are weaker than socially optimal outcome.
Mechanisms such as remanufacturing subsidies, cost-sharing schemes for take-back programs, or regulatory incentives that reward recovery activity can help align decentralized decisions with socially optimal outcomes.

\section{Numerical Illustration and Discussion}\label{sec:numerical}

This section provides numerical illustrations of the analytical results. 
Using the closed-form equilibrium expressions derived earlier, we examine how key parameters including consumer behavior, remanufacturing value, and the degree of product substituitability shape equilibrium prices, take-back incentives, and return volumes.

\subsection{Baseline equilibrium outcomes}

We begin with a symmetric baseline parameterization:
\[
\bar d = 200,\quad 
\alpha = 4,\quad 
c = 20,\quad 
v = 60,\quad 
k = 10,\quad 
\theta = 0.3,\quad 
\beta = 1.2,\quad 
\gamma_r = 10,\quad 
\varepsilon = 0.4.
\]
These values satisfy all feasibility conditions and yield interior equilibria under both proportional and inertia--responsiveness allocation mechanisms.

Table~\ref{tab:baseline} reports the corresponding symmetric equilibrium outcomes.
Under proportional allocation, take-back bonuses optimally collapse to zero, confirming the free-riding outcome predicted by Theorem~\ref{thm:symmetric}(a). 
Under the inertia--responsiveness mechanism, a strictly positive bonus emerges endogenously, leading to lower wholesale prices and higher return volumes. 
This comparison illustrates how behavioral frictions break the zero-bonus equilibrium and improve circular performance, as established in Theorem~\ref{thm:symmetric}(b).

\begin{table}[h!]
\centering
\begin{tabular}{lcccc}
\toprule
\textbf{Allocation rule} & $w^*$ & $p^*$ & $b^*$ & $Q_r$ \\
\midrule
Proportional ($s_i=\tfrac{1}{2}$) 
& 49.17 & 63.17 & 0.00 & 31.50 \\
Inertia--Responsiveness (\ref{eq:allocation}) 
& 47.65 & 62.45 & 8.33 & 36.90 \\
\bottomrule
\end{tabular}
\vspace{2mm}
\caption{Baseline equilibrium outcomes under alternative return-allocation mechanisms}
\label{tab:baseline}
\end{table}

\subsection{Behavioral drivers of take-back incentives}

Figure~\ref{fig:panel1} illustrates how consumer behavior shapes equilibrium take-back incentives.
The left panel plots the equilibrium bonus $b^*$ as a function of consumer responsiveness~$\beta$. 
A sharp threshold is observed at $\beta(v-k)=2\gamma_r$, below which the equilibrium degenerates to $b^*=0$. 
Beyond this threshold, bonuses increase linearly in~$\beta$, exactly as characterized in Theorem~\ref{thm:symmetric}(b).
The right panel shows the equilibrium bonus surface $b^*(\beta,\gamma_r)$ over a two-dimensional parameter grid. 
The boundary $\beta(v-k)=2\gamma_r$ clearly separates inactive and active take-back regimes. 
Within the active region, bonuses increase with responsiveness and decrease with inertia, confirming that higher behavioral frictions weaken take-back incentives.

\begin{figure}[h]
\centering
\begin{subfigure}[b]{0.48\textwidth}
    \centering
    \adjustbox{height=5cm}{%
        \includegraphics{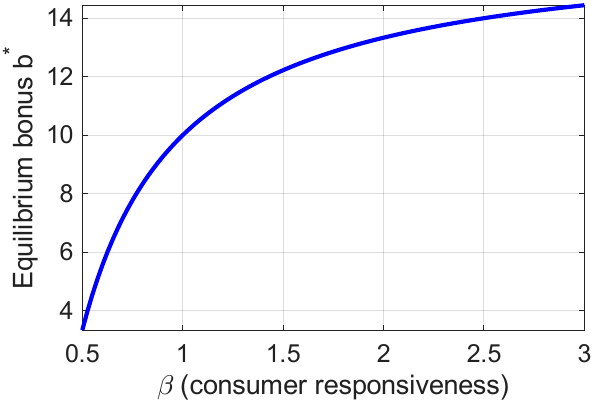}
    }
\end{subfigure}
\hfill
\begin{subfigure}[b]{0.48\textwidth}
    \centering
    \adjustbox{height=5cm}{%
        \includegraphics{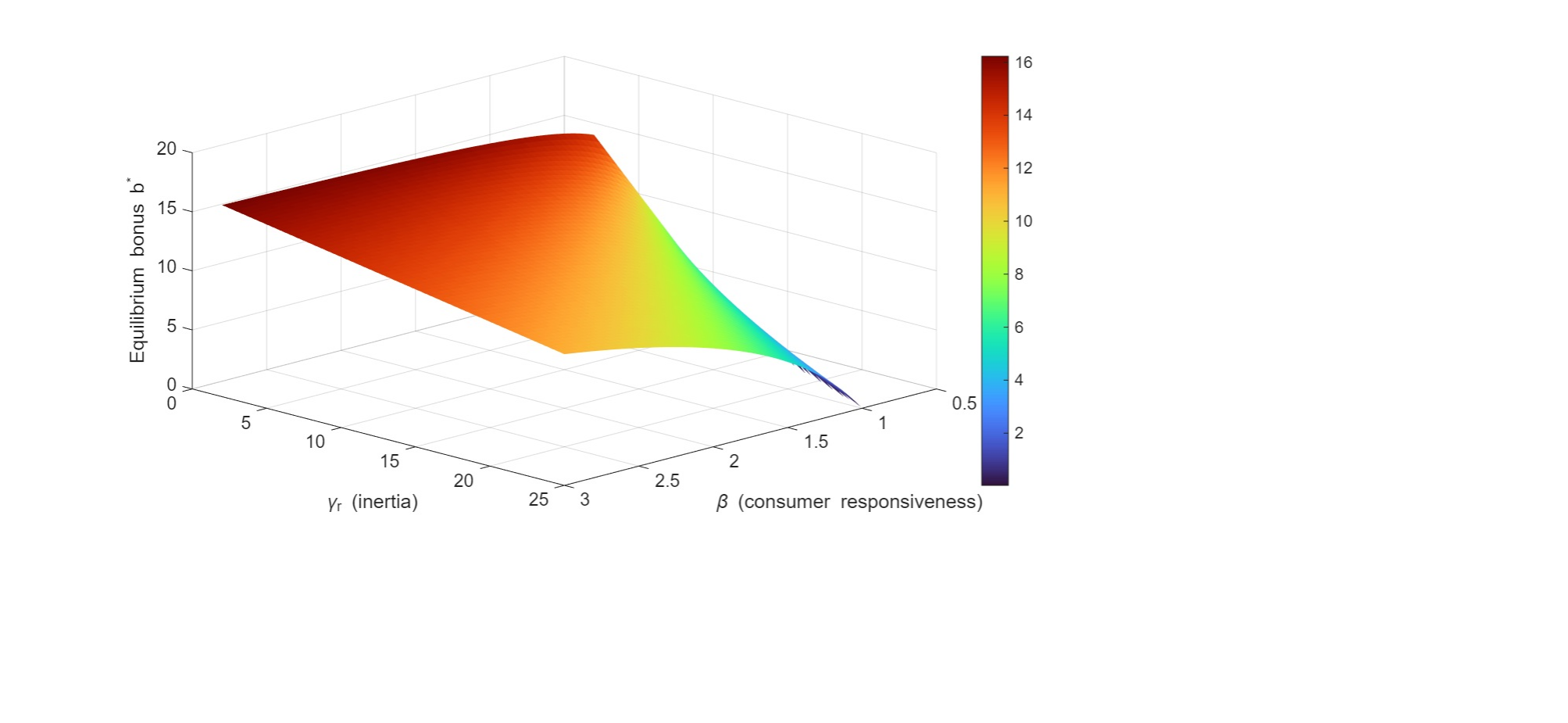}
    }
\end{subfigure}
\caption{Equilibrium take-back incentives.
\emph{Left}: Bonus $b^*$ as a function of consumer responsiveness $\beta$.
\emph{Right}: Bonus surface $b^*(\beta,\gamma_r)$ under inertia--responsiveness allocation.}
\label{fig:panel1}
\end{figure}

\subsection{Cross-market coupling and return outcomes}

Figure~\ref{fig:panel2} illustrates the interaction between forward pricing and reverse-channel outcomes.

The left panel shows the equilibrium wholesale price $w^*$ under proportional allocation as a function of remanufacturing value~$v$. 
The relationship is strictly decreasing and approximately linear, with slope $-\theta/2$, validating the cross-market coupling result in Theorem~\ref{thm:cross_market_proportional}. 
As remanufacturing becomes more profitable, manufacturers optimally reduce wholesale prices to expand sales and capture additional return volume.

The right panel plots the total return volume $Q_r(\beta,\gamma_r)$ under the inertia--responsiveness mechanism. 
Return quantities increase with consumer responsiveness and decrease with inertia, demonstrating that behavioral parameters play a central role in determining circular efficiency. 
High responsiveness and low inertia jointly deliver the largest recovery benefits.

\begin{figure}[h]
\centering
\begin{subfigure}[b]{0.48\textwidth}
    \centering
    \adjustbox{height=5cm}{%
        \includegraphics{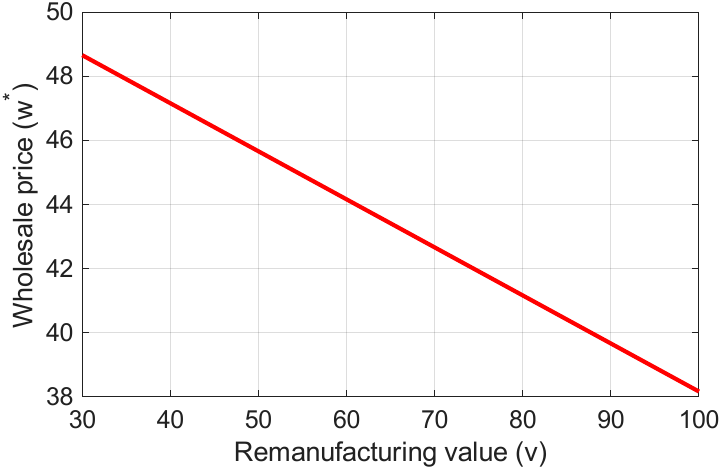}
    }
\end{subfigure}
\hfill
\begin{subfigure}[b]{0.48\textwidth}
    \centering
    \adjustbox{height=5cm}{%
        \includegraphics{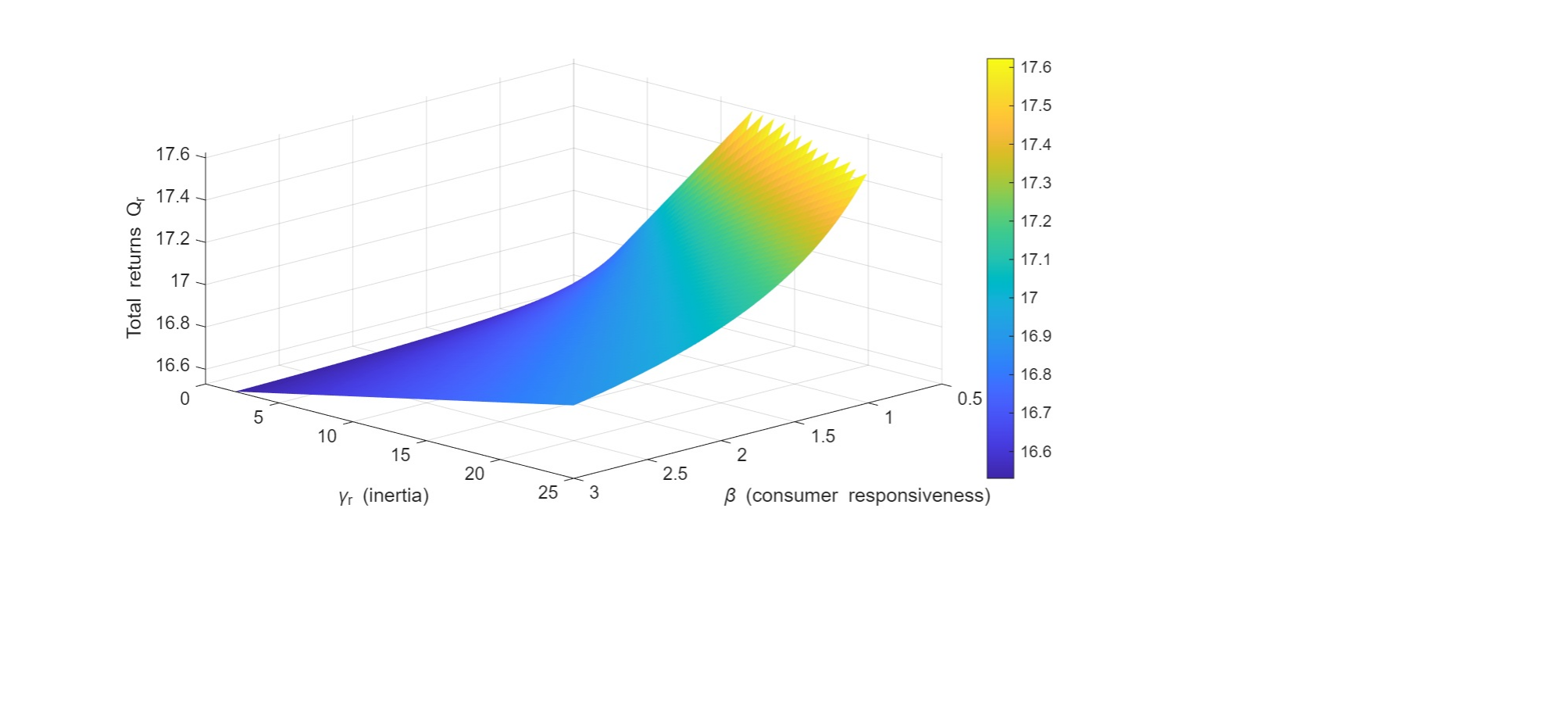}
    }
\end{subfigure}
\caption{Forward--reverse channel interactions.
\emph{Left}: Wholesale price $w^*$ as a function of remanufacturing value $v$.
\emph{Right}: Total return volume $Q_r(\beta,\gamma_r)$ under inertia--responsiveness allocation.}
\label{fig:panel2}
\end{figure}

\subsection{Effect of market competition (product substituitability)}\label{subsec:epsilon}

We now examine how the degree of market competition, captured by the cross-price
substituitability parameter $\varepsilon$, affects equilibrium outcomes.
Recall that $\varepsilon$ measures how easily consumers substitute between the two retailers:
low values correspond to weak competition and high brand loyalty, whereas values close to one represent strong substituitability and low loyality of customers towards retailers.

Using the closed-form symmetric equilibrium expressions derived in
Theorems~\ref{thm:symmetric} and~\ref{thm:cross_market_proportional},
we vary $\varepsilon$ from [0,1] while keeping all other parameters fixed
at their baseline values.

\textit{Wholesale and retail prices:}
Figure~\ref{fig:epsilon_prices} illustrates the effect of the substituitability parameter
$\varepsilon$ on equilibrium wholesale and retail prices.
As $\varepsilon$ increases, the products offered by the two retailers become more substitutable to consumers. In this regime, a price increase at one retailer primarily diverts demand toward the competing retailer rather than inducing a large loss of customers for that retailer. As a result, downstream price competition becomes less aggressive, allowing both retailers and manufacturers to sustain higher equilibrium prices despite intensified strategic interaction.
This explains the monotonic increase in both $w^*$ and $p^*$ as $\varepsilon$ rises.

\textit{Return volumes and circular performance:}
Figure~\ref{fig:epsilon_returns} illustrates how the degree of product substituitability $\varepsilon$
affects total return volumes $Q_r$.
Although $\varepsilon$ does not enter the return-allocation mechanism directly,
it influences return outcomes indirectly through its effect on equilibrium demand.
As $\varepsilon$ increases, stronger substitution weakens the disciplining effect of price competition, leading to higher equilibrium wholesale and retail prices. Depending on parameter values, these pricing adjustments can lead to higher realized return quantities even though aggregate demand continues to decrease with price.
\begin{figure}[h]
\centering
\begin{minipage}[t]{0.48\textwidth}
    \centering
    \includegraphics[height=4.5cm,keepaspectratio]{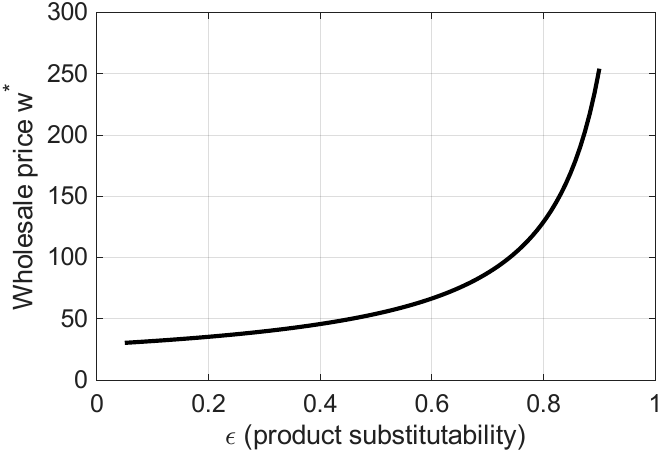}
\end{minipage}
\hfill
\begin{minipage}[t]{0.48\textwidth}
    \centering
    \includegraphics[height=4.5cm,keepaspectratio]{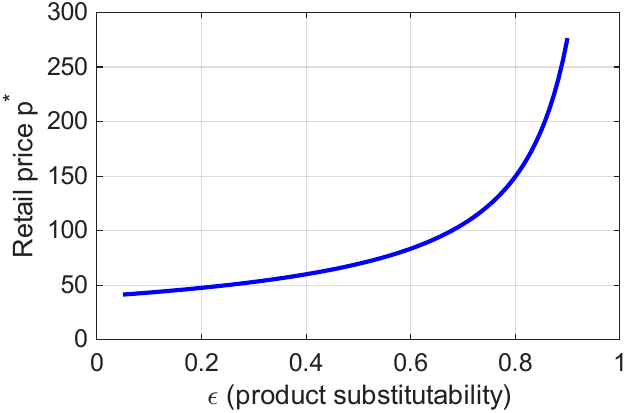}
\end{minipage}
\caption{Effect of product substitutability $\varepsilon$ on equilibrium prices.
\emph{Left}: Wholesale price $w^*$.
\emph{Right}: Retail price $p^*$.}
\label{fig:epsilon_prices}
\end{figure}

\begin{figure}[h]
\centering
\begin{minipage}{0.48\textwidth}
    \centering
    \includegraphics[height=4.5cm,keepaspectratio]{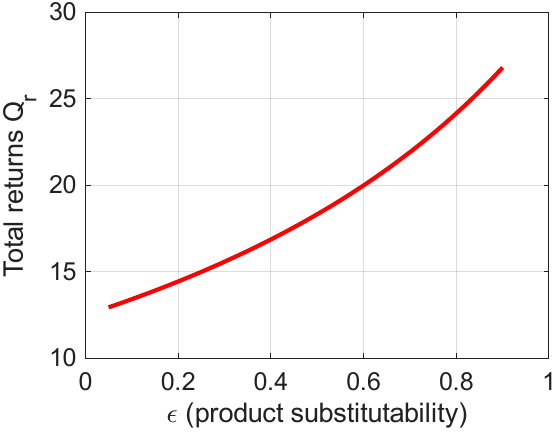}
\end{minipage}
\hfill
\begin{minipage}{0.48\textwidth}
    \centering
    \includegraphics[height=4.5cm,keepaspectratio]{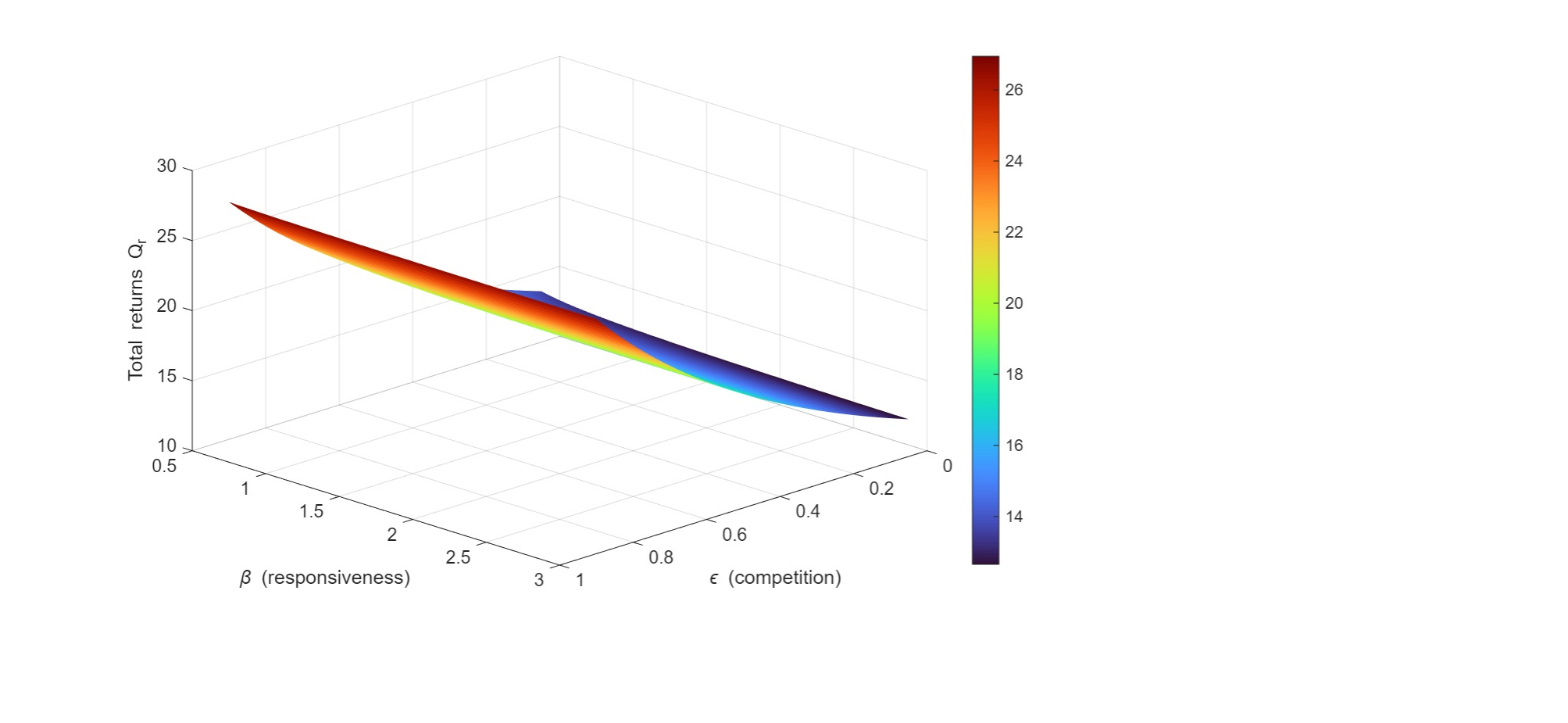}
\end{minipage}
\caption{Effect of competition on return outcomes.
\emph{Left}: Total return volume $Q_r$ as a function of $\varepsilon$.
\emph{Right}: Joint effect of product substitutability $\varepsilon$ and consumer responsiveness $\beta$ on total returns.}
\label{fig:epsilon_returns}
\end{figure}

\subsection{Discussion}

Taken together, the numerical experiments reflect the interaction between competitive pricing,
consumer behavior, and take-back incentives in a closed-loop supply chain.
The results show that consumer responsiveness and inertia play a central role in shaping equilibrium outcomes:
higher responsiveness amplifies take-back incentives and return volumes,
whereas greater inertia weakens manufacturers’ ability to stimulate returns through monetary bonuses.
At the same time, improvements in remanufacturing value feed back into forward-channel decisions,
leading manufacturers to strategically reduce wholesale prices in order to expand sales
and capture additional return flows.

The numerical analysis further reveals the role of market structure.
As the degree of product substitutability $\varepsilon$ increases, customers become less loyal to individual retailers and competition becomes increasingly redistributive across firms. This alters pricing incentives and equilibrium outcomes, even though total market demand may still respond to price changes.
In this regime, stronger substituitability alters pricing incentives by weakening competitive discipline, allowing firms to sustain higher equilibrium prices. While aggregate demand remains price-sensitive, these pricing effects can nonetheless translate into higher return quantities under certain parameter configurations.
Thus, stronger substitutability can enhance circular performance through increased return flows,
despite intensifying strategic interaction among firms.

\subsection{Managerial Implications}

The numerical analysis provides several insights that are directly relevant for managers/firms 
seeking to improve the performance of competitive CLSCs.

\begin{enumerate}[label=(\roman*),noitemsep]
\item \textbf{Incentive design matters:}
When consumer returns are allocated proportionally and do not respond to take-back bonuses,
manufacturers optimally free-ride and provide no incentives for product returns.
Thus take-back programs are unlikely to emerge voluntarily
unless consumer responsiveness is explicitly incorporated into program design.

\item \textbf{Behavioral frictions shape feasibility:}
Introducing consumer inertia and responsiveness generates a clear, quantifiable threshold
that determines when active take-back incentives become viable.
Higher consumer responsiveness strengthens the effectiveness of bonuses,
while greater inertia dampens their impact.
Managers must therefore account for behavioral heterogeneity when calibrating incentive levels,
as ignoring inertia can lead to overestimating the effectiveness of monetary rewards.

\item \textbf{Reverse profitability affects forward pricing:}
Increases in remanufacturing profitability through higher recovery value or return rates
lead manufacturers to strategically reduce wholesale prices.
By lowering prices upstream, firms expand sales volumes in order to capture additional returns,
demonstrating a strong interdependence between forward-channel pricing
and reverse-channel profitability.

\item \textbf{Market structure influences circular performance:}
Higher product substitutability weakens competitive pricing pressure and can lead to higher equilibrium prices. Through its indirect effect on pricing and sales incentives, this may result in larger return volumes, provided that appropriate take-back incentives are in place.
\end{enumerate}

Overall, these findings emphasize the need for joint consideration of pricing strategies, consumer behavior, market structure,
and remanufacturing incentives.

\subsection{Asymmetric Manufacturers: An Extension}

The closed-form equilibrium results derived earlier rely on parameter symmetry in order to obtain explicit joint solutions. We now extend the analysis to the asymmetric manufacturer case. Although a closed-form joint equilibrium is hard to compute, each manufacturer’s take-back bonus best response admits a complete analytical characterization with guaranteed uniqueness.

\textit{Reverse-channel objective:}
For manufacturer $i$, from \eqref{eq:manprofit} reverse-channel profit is
\[
\pi^{rev}_{M_i}
=
(v_i - b_i - k_i)\,
\theta D^{tot}(w)\,
s_i(b_1,b_2),
\]
where $D^{tot}(w)=D_1^*(w)+D_2^*(w)$ and the inertia--responsiveness allocation rule given in \eqref{eq:allocation}. Demand \eqref{eq:demand} does not depend on bonuses, so $D^{tot}(w)$ is constant with respect to $b_i$. Hence forward-channel profit terms are constant in $b_i$ and do not affect the best-response calculation for bonuses.

\textit{First-order condition:}
Differentiating with respect to $b_i$ gives
\[
\frac{\partial \pi^{rev}_{M_i}}{\partial b_i}
=
\theta D^{tot}(w)
\left[
- s_i
+
(v_i - b_i - k_i)
\frac{\partial s_i}{\partial b_i}
\right].
\]
An interior optimum therefore satisfies
\begin{equation}\label{eqn_foc_opt}
- s_i
+
(v_i - b_i - k_i)
\frac{\partial s_i}{\partial b_i}
= 0.
\end{equation}

Let
\[
B = \beta_1 b_1 + \beta_2 b_2 + 2\gamma_r,
\qquad
C_i = \beta_j b_j + 2\gamma_r, \; j \neq i.
\]

Direct differentiation yields
\[
\frac{\partial s_i}{\partial b_i}
=
\frac{\beta_i C_i}{B^2}
> 0,
\qquad
\frac{\partial^2 s_i}{\partial b_i^2}
=
-\frac{2\beta_i^2 C_i}{B^3}
< 0.
\]

Substituting into \eqref{eqn_foc_opt} gives
\[
\frac{\beta_i b_i + \gamma_r}{B}
=
(v_i - b_i - k_i)
\frac{\beta_i C_i}{B^2}.
\]

Multiplying by $B^2$ yields
\[
(\beta_i b_i + \gamma_r)(\beta_i b_i + C_i)
=
(v_i - b_i - k_i)\beta_i C_i.
\]

\textit{Quadratic best-response equation:}
Expanding and collecting terms produces a quadratic equation in $b_i$:
\begin{equation}\label{eqn_quadratic}
\beta_i^2 b_i^2
+
\beta_i(2C_i+\gamma_r)b_i
+
\big[\gamma_r C_i - \beta_i C_i (v_i-k_i)\big]
=
0.
\end{equation}

Thus, each manufacturer’s bonus best response is characterized by the real roots of this quadratic equation. Since the feasible set is $b_i \ge 0$, the economically admissible best response is the nonnegative root of \eqref{eqn_quadratic}, when it exists.

\textit{Uniqueness of the interior solution:}
Although the best-response equation is quadratic, the optimal solution is unique because the reverse-channel objective is strictly concave in $b_i$ over the convex feasible set $b_i \ge 0$. Specifically,
\[
\frac{\partial^2 \pi^{rev}_{M_i}}{\partial b_i^2}
=
\theta D^{tot}(w)
\left[
-2 \frac{\partial s_i}{\partial b_i}
+
(v_i - b_i - k_i)
\frac{\partial^2 s_i}{\partial b_i^2}
\right].
\]
Since $\partial s_i/\partial b_i > 0$, $\partial^2 s_i/\partial b_i^2 < 0$, and $v_i-k_i>0$ under reverse-channel viability, both terms are strictly negative, implying
\[
\frac{\partial^2 \pi^{rev}_{M_i}}{\partial b_i^2} < 0.
\]
Hence the objective is strictly concave and admits at most one maximizer. Therefore, even though \eqref{eqn_quadratic} has two algebraic roots, at most one satisfies the feasibility constraint and first-order condition, yielding a unique global maximizer. Moreover, because $C_i = \beta_j b_j + 2\gamma_r$, the best-response bonus is single-valued and continuous in the rival’s bonus decision.

\textit{Firm-specific incentive threshold:}
Evaluating the first-order derivative at $b_i=0$ yields
\[
\left.\frac{\partial \pi^{rev}_{M_i}}{\partial b_i}\right|_{b_i=0}
\propto
\beta_i (v_i-k_i) - C_i.
\]
Thus, a strictly positive take-back bonus is optimal if and only if
\[
\beta_i (v_i-k_i) > \beta_j b_j + 2\gamma_r,
\]
which generalizes the symmetric feasibility threshold derived earlier. If this condition does not hold, the derivative at zero is nonpositive and, since the objective is strictly concave in $b_i$, the constrained optimum occurs at the boundary $b_i = 0$. Incentive provision therefore depends on firm-specific remanufacturing margins and behavioral responsiveness.
Incentive provision therefore depends on firm-specific remanufacturing margins and behavioral responsiveness.

\textit{Comparative statics:}
Let $F(b_i;v_i,k_i)=0$ denote the quadratic best-response equation in \eqref{eqn_quadratic}. Since
\[
\frac{\partial F}{\partial b_i}
=
2\beta_i^2 b_i + \beta_i(2C_i+\gamma_r)
> 0,
\]
the implicit function theorem applies. Differentiating \eqref{eqn_quadratic} gives
\[
F_{v_i} = -\beta_i C_i < 0,
\qquad
F_{k_i} = \beta_i C_i > 0.
\]
Hence,
\[
\frac{\partial b_i^{BR}}{\partial v_i}
=
-\frac{F_{v_i}}{F_{b_i}}
> 0,
\qquad
\frac{\partial b_i^{BR}}{\partial k_i}
=
-\frac{F_{k_i}}{F_{b_i}}
< 0.
\]

Therefore, higher remanufacturing value strengthens take-back incentives, while higher processing cost weakens them. Within the interior-incentive regime, greater consumer responsiveness $\beta_i$ increases the marginal return-share gain from incentives and raises the optimal bonus level. Because each manufacturer’s best-response bonus is unique and continuous in the rival’s bonus, a Nash equilibrium in take-back bonuses exists by standard fixed-point arguments.

\section{Conclusion}\label{sec:conclusion}

In this paper, we developed a  two-echelon Stackelberg–Nash framework to study competition in a dual-channel CLSC with competing manufacturer–retailer pairs. By jointly modeling forward pricing decisions and reverse-channel take-back incentives within a unified analytical setting, we derived closed-form retailer pricing responses, manufacturer optimality conditions, and symmetric equilibrium outcomes under alternative return-allocation mechanisms.
We obtained several interesting insights. First, when returned products are allocated proportionally and consumer behavior does not respond to incentives, manufacturers optimally free-ride on each other’s collection efforts, leading to a collapse of take-back incentives. Second, introducing consumer inertia and responsiveness fundamentally alters this outcome by generating a clear and intuitive threshold for the emergence of positive take-back bonuses. Thus behavioral frictions play critical role in sustaining market-based return programs. Third, improvements in remanufacturing profitability or return rates feed back into the forward channel, inducing manufacturers to lower wholesale prices in order to expand sales and capture additional returns. This forward–reverse coupling demonstrates that pricing and recovery decisions cannot be analyzed in isolation in competitive closed-loop systems.
Numerical experiments validate the analytical findings and illustrate how consumer behavior, market competition, and remanufacturing value jointly shape equilibrium prices, incentives, and return volumes. Together, the results provide a coherent explanation of when voluntary take-back programs are viable and how reverse-channel profitability reshapes competitive pricing strategies.
Several  avenues for future research emerge from this work. Extending the framework to dynamic or multi-period settings would allow the study of learning, reputation and incentive design. Incorporating uncertainty in return quantity or quality would further make the model more realistic. Relaxing symmetry assumptions to allow for heterogeneous firms, products, or consumer segments would broaden applicability to more complex competitive environments.
\bibliographystyle{splncs04}   
\bibliography{ref}
\appendix

\section{Proof of Theorem~\ref{thm:retailer} (Retailer-Stage Uniqueness)}
\label{app:retailer_proof}

Fix wholesale prices $(w_1,w_2)$ chosen by the manufacturers in Stage~1.  
In Stage~2, retailers simultaneously choose retail prices $(p_1,p_2)$ to maximize their respective profits
\[
\pi_{R_i}(p_i,p_j;w_i)=(p_i-w_i)D_i(p_i,p_j)-O_{R_i}, \qquad i\neq j.
\]

Using the demand functions defined in \eqref{eq:demand}, the first-order condition for retailer~$i$ is
\[
\frac{\partial \pi_{R_i}}{\partial p_i}
=\bar d_i+\alpha_i w_i+\varepsilon\alpha_j p_j-2\alpha_i p_i=0,
\]
which yields the system of equations reported in (R--FOC):
\[
\bar d_1+\alpha_1 w_1+\varepsilon\alpha_2 p_2-2\alpha_1 p_1=0,\qquad
\bar d_2+\alpha_2 w_2+\varepsilon\alpha_1 p_1-2\alpha_2 p_2=0.
\]

This linear system can be written compactly as
\[
Mp=h,
\]
where
\[
M=
\begin{pmatrix}
2\alpha_1 & -\varepsilon\alpha_2\\
-\varepsilon\alpha_1 & 2\alpha_2
\end{pmatrix},
\qquad
p=\begin{pmatrix}p_1\\p_2\end{pmatrix},
\qquad
h=\begin{pmatrix}\bar d_1+\alpha_1 w_1\\ \bar d_2+\alpha_2 w_2\end{pmatrix}.
\]

The determinant of $M$ is
\[
\det(M)=\alpha_1\alpha_2(4-\varepsilon^2)>0
\quad \text{for } \varepsilon\in[0,1],
\]
which guarantees that $M$ is nonsingular. Hence the system admits a unique solution
$p=M^{-1}h$, yielding the equilibrium retail prices given in
\eqref{eq:p1star_main}--\eqref{eq:p2star_main}.

Finally, note that
\[
\frac{\partial^2 \pi_{R_i}}{\partial p_i^2}=-2\alpha_i<0,
\]
so each retailer’s profit function is strictly concave in its own price.  
Therefore, the solution obtained above is the unique Nash equilibrium of the retailer pricing subgame.  
\hfill $\square$

\section{Substituted Demands and Derivatives}
\label{app:sub}

Substituting the equilibrium retail prices $p_i^*(w_1,w_2)$ from
\eqref{eq:p1star_main}--\eqref{eq:p2star_main} into the demand functions
\eqref{eq:demand}, and defining $\Delta=4-\varepsilon^2$, we obtain the reduced-form demands faced by manufacturers:
\[
D_1^*(w_1,w_2)
=\frac{1}{\Delta}
\bigl(2\bar d_1+\varepsilon\bar d_2
-\alpha_1(2-\varepsilon^2)w_1+\varepsilon\alpha_2 w_2\bigr),
\]
\[
D_2^*(w_1,w_2)
=\frac{1}{\Delta}
\bigl(\varepsilon\bar d_1+2\bar d_2
+\varepsilon\alpha_1 w_1-\alpha_2(2-\varepsilon^2)w_2\bigr).
\]

Differentiating with respect to wholesale prices yields constant marginal effects:
\[
\frac{\partial D_1^*}{\partial w_1}
=-\frac{\alpha_1(2-\varepsilon^2)}{\Delta}<0,
\qquad
\frac{\partial D_1^*}{\partial w_2}
=\frac{\varepsilon\alpha_2}{\Delta}>0,
\]
with analogous expressions for $D_2^*$.  
These derivatives are used directly in the manufacturers’ first-order conditions \eqref{eqn_1}.  
\hfill $\square$

\section{Proof of Theorem~\ref{thm:symmetric} (Symmetric Equilibrium)}
\label{app:symmetric_proof}

Assume symmetry:
$\bar d_i=\bar d$, $\alpha_i=\alpha$, $c_i=c$, $v_i=v$, $k_i=k$, and $\beta_i=\beta$.
Under symmetry, Theorem~\ref{thm:retailer} implies
\[
p=\frac{\bar d/\alpha+w}{2-\varepsilon},
\qquad
D(w)=\frac{\bar d-\alpha(1-\varepsilon)w}{2-\varepsilon}.
\]

\paragraph{(a) Proportional allocation.}
Under proportional allocation, each manufacturer receives a fixed share
$s_i=\tfrac12$ (see \eqref{eq:allocation}) of total returns. Manufacturer profit (\eqref{eq:manprofit}) reduces to 
\[
\pi_M(w,b)
=(w-c)D(w)+\theta(v-b-k)D(w)-O_M.
\]

Differentiating the above w.r.t $b$  gives
\[
\frac{\partial \pi_M}{\partial b}=-\theta D(w)<0,
\]
implying that the optimal bonus is $b^*=0$ (as we have $D(w) > 0$ under operating conditions by \textbf{A.2}).

Substituting $b=0$, the manufacturer’s problem reduces to maximizing
\begin{eqnarray}\label{eqn_manu_w}
 \pi_M(w)=[(w-c)+\theta(v-k)]D(w).   
\end{eqnarray}

Using $D'(w)=-\alpha(1-\varepsilon)/(2-\varepsilon)$ , the first-order condition
$\partial \pi_M/\partial w=0$ yields
\[
w^*=\frac{\bar d}{2\alpha(1-\varepsilon)}+\frac{c-\theta(v-k)}{2},
\]
as stated in the theorem.

\paragraph{(b) Inertia--responsiveness allocation.}
Under the inertia--responsiveness allocation rule~\eqref{eq:allocation},
\[
s_i(b_1,b_2)=\frac{\beta b_i+\gamma_r}{\beta b_1+\beta b_2+2\gamma_r},
\qquad i\in\{1,2\}.
\]
In a symmetric equilibrium, manufacturers choose identical bonuses
$b_1=b_2=b$, which implies
\[
s_i=\frac{\beta b+\gamma_r}{2(\beta b+\gamma_r)}=\frac{1}{2}.
\]

Differentiating $s_i$ with respect to $b_i$ gives
\[
\frac{\partial s_i}{\partial b_i}
=\frac{\beta(\beta b_j+2\gamma_r)}
{(\beta b_i+\beta b_j+2\gamma_r)^2},
\qquad j\neq i.
\]
Evaluating this derivative under symmetry $(b_i=b_j=b)$ yields
\[
\left.\frac{\partial s_i}{\partial b_i}\right|_{\mathrm{sym}}
=\frac{\beta(\beta b+2\gamma_r)}{4(\beta b+\gamma_r)^2}.
\]

The manufacturer’s first-order condition with respect to the take-back
bonus $b_i$ (see~\eqref{eqn_2}) is
\[
-s_i+(v-b_i-k)\frac{\partial s_i}{\partial b_i}=0.
\]
Substituting $s_i=\tfrac12$ in the expression above gives
\[
-\frac12
+(v-b-k)\frac{\beta(\beta b+2\gamma_r)}
{4(\beta b+\gamma_r)^2}=0.
\]
Multiplying both sides by $4(\beta b+\gamma_r)^2$ and simplifying yields
\[
2(\beta b+\gamma_r)^2
=\beta(v-b-k)(\beta b+2\gamma_r),
\]
which reduces to
\[
3\beta b=\beta(v-k)-2\gamma_r.
\]
Hence, the  stationary point is
\[
b^*=\frac{v-k}{3}-\frac{2\gamma_r}{3\beta}.
\]

The equilibrium bonus satisfies $b^*>0$ if and only if
$\beta(v-k)>2\gamma_r$.
If this condition fails, the objective function is strictly decreasing
at $b=0$, and the constrained optimum is $b^*=0$.

Keeping wholesale prices fixed, the manufacturer’s profit as a function
of $b_i$ can be written as (see \eqref{eq:manprofit})
\[
\pi_M(b_i)=\theta D^*(w)\,(v-b_i-k)s_i(b_i,b_j)
+\text{terms independent of }b_i.
\]
Differentiating twice,
\[
\frac{\partial^2 \pi_M}{\partial b_i^2}
=\theta D^*(w)
\left[
-2\frac{\partial s_i}{\partial b_i}
+(v-b_i-k)\frac{\partial^2 s_i}{\partial b_i^2}
\right].
\]
Under inertia--responsiveness allocation (see \eqref{eq:allocation}),
\[
\frac{\partial^2 s_i}{\partial b_i^2}
=-\frac{2\beta^2(\beta b_j+2\gamma_r)}
{(\beta b_i+\beta b_j+2\gamma_r)^3}<0.
\]
Since $D^*(w)>0$ by Assumption~\textbf{A.2}, it follows that
\[
\left.\frac{\partial^2 \pi_M}{\partial b_i^2}\right|_{\mathrm{sym}}<0,
\]
so $b^*$ is a strict global maximizer of the manufacturer’s bonus-choice
problem.
Substituting $b^*$ into the manufacturer’s wholesale-price
first-order condition \eqref{eqn_manu_w} yields the equilibrium wholesale price $w^*$
 in Theorem~\ref{thm:symmetric}.  
\qed

\section{Proof of Theorem \ref{thm:cross_market_proportional} (Forward--reverse coupling under proportional allocation)}\label{app:cross_market_proof}
 From Theorem \ref{thm:symmetric}(a), the symmetric equilibrium wholesale price under proportional allocation is given by
\[
w^* = \frac{\bar d}{2\alpha(1-\varepsilon)} + \frac{c-\theta(v-k)}{2}.
\]
Differentiating $w^*$ with respect to $v$ while holding all other parameters fixed yields
\[
\frac{\partial w^*}{\partial v} = -\frac{\theta}{2}.
\]
Similarly, differentiating with respect to $\theta$ yields
\[
\frac{\partial w^*}{\partial \theta} = -\frac{v-k}{2}.
\]
Since $\theta>0$ and $v-k>0$ by assumption, both derivatives are strictly negative, establishing the result.
\qed

\section{Proof of Corollary \ref{cor:coupling_inertia} (Forward--reverse coupling under inertia--responsiveness)}\label{app_coupling}

From Theorem~\ref{thm:symmetric}(b), the symmetric equilibrium bonus is
\[
b^*=\frac{v-k}{3}-\frac{2\gamma_r}{3\beta},
\]
and the equilibrium wholesale price is
\[
w^*=\frac{\bar d}{2\alpha(1-\varepsilon)}+\frac{c-\theta(v-b^*-k)}{2}.
\]
Substituting $b^*$ into the expression for $w^*$ and differentiating with respect to $v$ while holding all other parameters fixed yields
\[
\frac{\partial w^*}{\partial v}=-\frac{\theta}{3}.
\]
Differentiating with respect to $\theta$ yields
\[
\frac{\partial w^*}{\partial \theta}=-\frac{v-b^*-k}{2}.
\]
Under the interior condition $\beta(v-k)>2\gamma_r$, we have $b^*>0$ and $v-b^*-k>0$, implying both derivatives are strictly negative.
\qed

\section{Proof of Theorem~\ref{thm:welfare} (Welfare Inefficiency)}
\label{app:welfare_proof}

Since $b \ge 0$ and $k > 0$, it follows that
\[
\theta v > \theta (v - b - k)
\]
whenever $\theta v > 0$. Let
\[
F(w) = D(w) + \bigl(w - c + \theta (v - b - k)\bigr) D'(w)
\]
denote the first-order condition under decentralized decision-making, and
\[
G(w) = D(w) + \bigl(w - c + \theta v\bigr) D'(w)
\]
denote the corresponding first-order condition under the social planner’s problem.

From \eqref{eqn_sym_pr_dem}, the demand function $D(w)$ is linear and strictly
decreasing in $w$, with
\[
D'(w) = -\frac{\alpha(1-\varepsilon)}{2-\varepsilon} < 0
\quad \text{and} \quad
D''(w)=0 .
\]
Therefore,
\[
F'(w) = D'(w) + D'(w) + \bigl(w-c+\theta(v-b-k)\bigr)D''(w)
      = 2D'(w) < 0,
\]
and similarly,
\[
G'(w) = D'(w) + D'(w) + \bigl(w-c+\theta v\bigr)D''(w)
      = 2D'(w) < 0.
\]
Hence, both $F(w)$ and $G(w)$ are strictly decreasing functions of $w$, and each
first-order condition admits a unique solution.

Because $\theta v > \theta (v - b - k)$ and $D'(w)<0$, it follows that
\[
G(w) < F(w) \qquad \text{for all } w .
\]
Since both functions are strictly decreasing, the equation $G(w)=0$ is satisfied
at a strictly lower wholesale price than the equation $F(w)=0$. That is, the equilibrium wholesale price manufacturer under the decentralized setting $w^{*}_{D}$ and under social planner $w^{*}_{SP}$ satisfy,
\[
w^{*}_{SP} < w^{*}_{D}.
\]

Consequently, the socially optimal wholesale price under social planner is strictly lower than the
decentralized equilibrium wholesale price. Since $D(w)$ is strictly decreasing in $w$ from \eqref{eqn_sym_pr_dem} ,
this implies strictly higher equilibrium demand and, hence, strictly higher
return volumes under the social optimum.
\hfill $\square$


\end{document}